\documentclass[conference]{IEEEtran}

\ifCLASSINFOpdf

\else

\fi

\usepackage{amsmath}
\usepackage[normalem]{ulem}
\usepackage{soul} 
\usepackage{paralist}

\usepackage{amssymb}
\usepackage{graphicx}
\usepackage{epstopdf}
\usepackage{subfig}
\usepackage{multirow}
\usepackage{enumerate}
\usepackage{algorithm}
\usepackage{algorithmicx}
\usepackage{algpseudocode}

\usepackage{xcolor}

\newtheorem{lemma}{Lemma}

\def\F{\mathcal{F}}
\def\N{\mathcal{N}}

\def\Prob{\mathbb{P}}

\def\etc{\emph{etc. }}
\def\eg{\emph{e.g. }}
\def\ie{\emph{i.e. }}

\IEEEoverridecommandlockouts

\begin{document}
%
\title{Proportional Fair Rate Allocation for Private Shared Networks\thanks{This work was supported by Science Foundation Ireland under Grants 11/PI/1177 and 13/RC/2077}}


\author{\IEEEauthorblockN{Saman Feghhi,\ Douglas J. Leith,\ Mohammad Karzand}
\IEEEauthorblockA{School of Computer Science and Statistics, Trinity College Dublin\\
Email: \{feghhis,doug.leith,karzandm\}@tcd.ie}}

\maketitle

\IEEEpeerreviewmaketitle

\begin{abstract}
In this paper, we consider fair privacy in a shared network subject to traffic analysis attacks by an eavesdropper.  We initiate the study of the joint trade-off between privacy, throughput and delay in such a shared network as a utility fairness problem and derive the proportional fair rate allocation for networks of flows subject to privacy constraints and delay deadlines.  

\end{abstract}

\section{Introduction}
\label{sec:intro}
\noindent
Privacy in packet switched networks has attracted increasing interest in the research community over the last decade. Traditionally encryption techniques have been used to provide privacy by hiding the content of communication messages from an eavesdropper.  However, more sophisticated attacks have been developed in recent years that use characteristics other than the content of the packets to gain information about the content of transmissions.  For example, it has been demonstrated that packet length, count and the inter arrival times between packets can be used to infer with high accuracy the identity of web pages being browsed {\cite{feghhi15}}, video being watched and the content of encrypted voice calls {\cite{white11,wright08}}.	For attacks which make use of packet length information the obvious defence is to pad/fragment packets to be of fixed size.  However, this simple defence is insufficient to protect against timing only attacks, since these make no use of packet size information.  Note that timing only attacks can be powerful e.g. in 
{\cite{feghhi15,feghhi16}} it is demonstrated that the web site being browsed can be successfully inferred with greater than 90\% accuracy using timing information alone.   

Timing attacks make use of the properties of the packet \emph{stream}, rather than of individual packets, and defence against such attacks therefore requires use of traffic shaping to modify the timing pattern of the stream of packets transmitted over the network such that it becomes hard for an eavesdropper to learn about the original message.   Such traffic shaping can be achieved by inserting dummy packets into the packet stream to mask idle periods, by delaying/buffering packets to modify their timing and of course by dropping packets.  In practice this shaping might be performed by, for example, end hosts or by a VPN gateway. 

\begin{figure}
\centering
\includegraphics[width=0.6\columnwidth]{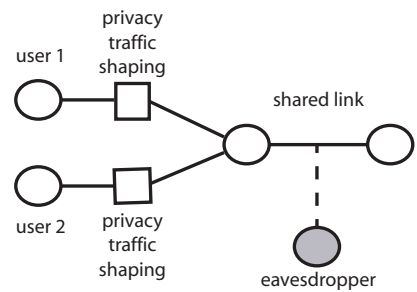}
\caption{Schematic illustrating simple two flow example.}\label{fig:intro}
\end{figure}

However, such privacy enhancing traffic shaping can impose a cost on the user and on the network.  For example, insertion of dummy packets increases the load on the network while buffering user packets for longer increases delay but can reduce the need for dummy packets to ensure privacy.   Importantly, for users sharing common a network path this creates a joint trade-off between their privacy, throughput and delay.  This is illustrated in Figure \ref{fig:intro} which shows two users whose traffic is shaped to enhance privacy and is then sent over a shared link which may be subject to eavesdropping.   An increase in the rate of dummy packet transmissions by user 1 reduces the available bandwidth for user 2 on the shared link, and so to maintain privacy may require user 2 to reduce the rate at which useful (non-dummy) packets are sent, to buffer packets for longer (to disrupt timing without adding extra dummy packets) or to drop more packets (again to disrupt timing patterns).  Alternatively, user 2 may choose to sacrifice privacy in order to avoid reducing throughout, increasing delay \etc   That is, increased privacy for one user may come at the cost of decreased throughput, increased delay and/or reduce privacy for another user.  

In this paper, we initiate the study of the joint trade-off between privacy, throughput and delay in a shared network as a utility fairness problem.   We derive the proportional fair rate allocation for networks of flows subject to privacy constraints and delay deadlines.  To the best of our knowledge this is the first study of \emph{fair} privacy in a shared network subject to timing attacks.


\section{Related Work}
\label{sec:related}
\noindent
{There exists a large body of research focused on fair rate allocation in shared networks but without consideration of privacy, see for example \cite{eryilmaz05,massoulie99,banchs07} and references therein.   With regard to privacy, fairness has previously been considered in the context of mix networks \cite{chaum81} where the aim is to achieve unlinkability/anonymity.  Mishra et al. \cite{mishra12} propose a proportional fair scheduling model that preserves source anonymity and in \cite{mishra15} investigate the trade-off between anonymity and quality of service. The focus in both papers is on hiding the identity of the users in a shared network.}   With regard to traffic analysis and associated defences, previous studies have established that the packet size, count, transmission rate, inter-arrival times \etc can reveal significant information about a traffic flow, see e.g. \cite{white11,dyer12,feghhi15,feghhi16} and references therein.  A number of defences have been proposed to counter these traffic analysis attacks, mostly addressing attacks that make use of packet size, count and the length of transmission \eg they modify the traffic by padding the packets to be the same length or by adding dummy traffic to the end of the packet sequence to mask the packet count \cite{guan01}.  These defences are, however, ineffective against attacks that solely use timing information of network flows, see for example \cite{feghhi15} and \cite{feghhi16}.  To the best of our knowledge none of this previous work on defences against traffic analysis attacks has considered fairness or the joint aspect of the trade-off between privacy, throughput and delay in a shared network. 

\section{Attack Model \& Privacy for a Network Flow}
\label{sec:notionprivacy}
\noindent
We consider a network path where time $t\in\{0,1,\cdots\}$ is slotted and in slot $t$ a new packet may arrive for transmission across the network. We let random variable $X(t)=1$ if a packet arrives in slot $t$ and $X(t)=0$ otherwise.    It is assumed that packets themselves contain no interesting information for an eavesdropper, which means that the packets are strongly encrypted and are of constant size.  However, the sequence $\mathbf{X}:=\{X(t)\}$  of packet arrivals contains information that can be used to reveal the characteristics of the traffic flow.   To protect this information, the arriving packet stream is passed through a traffic shaper before being transmitted across the network.   We let random variable $Y(t) = 1$ if a packet is transmitted across the path in slot $t$ and $Y(t) = 0$ otherwise.   

Our attack model is that the sequence $\mathbf{Y}:=\{Y(t)\}$ of transmissions is observable by an adversary, but not the sequence $\mathbf{X}$ of arrivals.  Note that in general we expect that an eavesdropper listening to network transmissions does not observe output sequence $\mathbf{Y}$ directly, but rather only after transformation by the MAC layer \ie after buffering, scheduling delay \etc   This will generally increase privacy (this follows from the data processing inequality, and more specifically \eg passing through a queue increases entropy).  Nevertheless, we take a conservative approach and assume a strong attacker who can perfectly invert these changes and so recover $\mathbf{Y}$.

We use the mutual information $I(\mathbf{Y}; \mathbf{X})$ between sequences $\mathbf{X}$ and $\mathbf{Y}$ as our measure of privacy, similarly to \cite{wu06} and others. Mutual information measures the capacity of the information channel between sequences $\mathbf{X}$ and $\mathbf{Y}$.  When $I(\mathbf{Y}; \mathbf{X})=0$ the sequences are statistically independent and more generally mutual information captures the exposure to watermarking attacks, which can be thought of as a worst case situation in our setup, namely where the user sequence $\mathbf{X}$ has a signature intentionally embedded within its timing pattern by an adversary (e.g. by a web site which is being downloaded)
Recall that $I(\mathbf{Y}; \mathbf{X}) = H(\mathbf{Y})-H(\mathbf{Y}|\mathbf{X})$, where $H(\mathbf{Y})$ is the entropy of the output sequence and $H(\mathbf{Y}|\mathbf{X})$ the conditional entropy between the output and input sequences.    Privacy is maximised when $I(\mathbf{Y}; \mathbf{X})$ is minimised (as already noted, when $I(\mathbf{Y}; \mathbf{X})=0$, the output sequence $\mathbf{Y}$ is statistically independent from the input sequence $\mathbf{X}$ and we say transmissions are fully private).  


\section{On-Off Traffic Shaping Policy}
\label{sec:policy}
\noindent
Transmitting a packet in every slot regardless of the pattern of the arriving user traffic ensures that the transmitted packet sequence contains no information about the user traffic\footnote{Formally,  suppose output sequence $Y(t)=1$, $t=1,2,\cdots$ \ie a packet is transmitted in every slot.   The entropy of the transmission at a single slot is $H(Y(t)) =  p\log p + (1-p)\log(1-p)$ where $p=Prob(Y(t)=1)$.  Since $p=1$ when a packet is always transmitted, $H(Y(t))=0$.  The entropy of the \emph{sequence} $\{Y(t)\}$ satisfies $H(\mathbf{Y}) \le \sum_t H(Y(t)) = 0$ and since $H(\mathbf{Y}) \ge 0$ it follows that the entropy of the sequence $\{Y(t)\}$ is zero.   Similarly, $H(\mathbf{Y}|\mathbf{X})=0$ and so the mutual information $I(\mathbf{Y}; \mathbf{X})=0$. } and $I(\mathbf{Y}; \mathbf{X})=0$.  Similarly, when no packets are transmitted, $Y(t)=0$, $t=1,2,\cdots$, then also trivially $I(\mathbf{Y}; \mathbf{X})=0$.   
However, transmitting a packet in every slot means that when there is no information packet to send a dummy packet must be transmitted and so is clearly wasteful.  And transmitting no packets at all is clearly private but not useful for communication.   
This motivates use of an on-off approach to transmission.  That is, we specify an interval $\tau$.  In the first slot of this interval we transmit a packet (sending a dummy packet if no information packets are available\footnote{Note that it is also possible to adopt a partially private approach where, when no information packet is available to send, a dummy packet is transmitted with some probability less than one.  In this case the mutual information $I(\mathbf{Y}; \mathbf{X})$ will be non-zero, but can be controlled by adjusting the probability with which dummy packets are sent.   However, we leave this more general case for future work.}), and over the remaining $\tau-1$ slots no packets are transmitted (any arriving information packets are buffered in a queue).  That is, $Y(t)=1$ for $t=1$ and $Y(t)=0$ for $t\in\{2,\cdots,\tau\}$.  It is known from queueing theory that deterministic service minimises the queueing delay for a specified mean service rate \cite{kingman62}, i.e. periodic service minimises delay.  Conveniently, periodic service also means that the pattern of packet transmissions contains no information, i.e. the mutual information between the transmitted sequence $\mathbf{Y}$ and the arrival process $\mathbf{X}$ is zero, $I(\mathbf{Y}; \mathbf{X})=0$ and so transmissions are fully private.   Hence this traffic shaping approach is a minimum delay maximum privacy one.  By adjusting the mean transmit rate via parameters $g$ and $\tau$ we can tune the trade-off between the number of dummy packets sent (which waste network bandwidth) and the buffering delay experienced by information packets.

%
%

This on-off traffic shaper can be modelled as a so-called Fixed Cycle Traffic Light (FCTL), first studied in\cite{newell60} in the context of vehicular traffic.   Consider cars arriving at a junction controlled by a FCTL. The light has two states which divides the total cycle into fixed length green (g) and red (r) cycles. In the red cycle, a new arrival enters a queue, waiting to cross the junction. When the light turns green, cars in the queue cross the junction one at a time until the cycle lasts or queue becomes empty.   Cars arriving at the junction during a green cycle and finding the queue empty proceed to cross the junction.    It is this latter characteristic which makes the FCTL different from a conventional queue with periodic service.   In our setup the green cycle is of duration $g=1$ slot and the red cycle is of duraton $r=\tau-1$ slots.  
The characteristics of a FCTL have been much studied, see for example the overview in \cite{darroch64}, with estimates of the average queue length given in \cite{newell60} and \cite{broek06}.
\subsection{Delay}
\noindent
Suppose the input sequence $\{X(t)\}$ is i.i.d with $\Prob(X(t)=1) = p$ and $\Prob(X(t)=0) = 1-p$.  That is, the traffic arrivals form a Bernoulli process. Given that the time is slotted, following \cite{newell60}, the average waiting time $w$ for each packet measured in time slots is  
\begin{align}
w = (\tau-g) (1-p)^{-1}\tau^{-1}(E(q_x) / p + (\tau - g + 1) / 2). \label{eq:delay}
\end{align}
where $E(q_x)$ denotes the length of queue right at the end of green (serving) period of the $x$th cycle and $\tau$ is the length of each cycle \ie $\tau = r + g$.

In order to calculate $w$ we need to find the equilibrium distribution for $q_x$ and evaluate $E(q_x)$. While the calculation of $q_1$ and $q_2$ is straightforward, the evaluation of $q_3$ onwards quickly becomes tedious and the expressions cumbersome.  Several attempts have therefore been made to derive accurate estimates of $E(q_x)$ which are simpler in nature.	Two estimates, that of Newell \cite{newell60} and of Miller \cite{broek06}, are illustrated in Figure \ref{fig:estimations} and compared against numerical simulation data.  In the sequel we use Miller's estimate due to its accuracy and simplicity.  This is given by
\begin{align}
E(q_x) \approx \max\{\frac{2p\tau - g}{2(g - p\tau)}.(1-p),0\}
\end{align}

\begin{figure}[!t]
\centering
	\includegraphics[width=0.8\columnwidth]{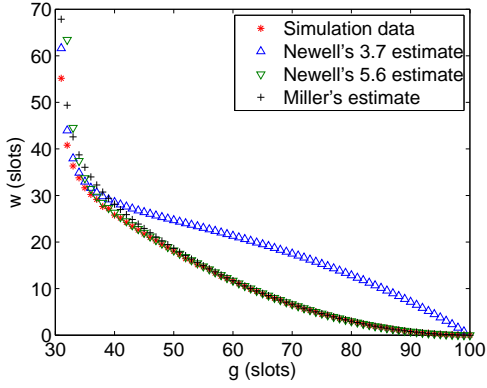}
	\caption{Comparison of Newell's (equations (3.7) and (5.6) in \cite{newell60}) and Miller's estimates for the waiting time at a FCTL with simulation data. The simulation values are averaged over 1000 cycles with $p=0.3$ and $\tau=100$. The data for $g/\tau < p$ is excluded as the queue is unstable in this region.}
	\label{fig:estimations}
\end{figure}


\subsection{Rate of Dummy Packet Transmissions}
\noindent
In addition to the delay introduced by traffic shaping we are also interested in the fraction $d$ of slots expended on dummy packet transmissions (when no information packet is available to send at a slot in an on cycle).   The information packet arrivals $\nu$ during an off period are distributed as:
\begin{align}
	Prob(\nu = k) = {T \choose k} p^k (1-p)^{T-k},
\end{align}
and the associated probability generating function is $K(z) = [(1-p)+pz]^T$.   Regarding $\Pi(z)$, the probability generating function for the limiting distribution $\pi_j$ the length of the queue at the start of the first slot of an on period, we have:
\begin{align}
	\Pi(z)  & = \frac{\sum_{j=0}^{g-1} \pi_{j}(z^{g}-z^j)}{z^{g}/K(z)-1} 
	 = \frac{\sum_{j=0}^{g} \pi_{j}(z^{g}-z^j)}{z^{g}[(1-p)+pz]^{-T}-1}.
\end{align}
We know $\Pi(1)=1$, which means after using L'hopital's rule
\begin{equation}
d = \frac{1}{\tau}\sum_{j=0}^{g-1} \pi_{j}(g-j) = \frac{g}{\tau} - p, \label{eq:dval}
\end{equation}
the average number of dummy packets transmitted over cycle.
	
We will assume that $g - p\tau>0$. This means that on average our service can accommodate all of the arrivals to the queue and the queue is stable. 	Stability can be seen by looking at the polynomial $z^{T_{on}}[(1-p)+pz]^{-T}-1$ which, for stability, should have no zeros in or on the unit circle.  Now,
\begin{align}
	z^{g} & = [(1-p)+pz]^{\tau},\ 
	g z^{g}  = p \tau [(1-p)+pz]^{\tau-1} 
\end{align}
Dividing the two equations, for $|z|<1$ we should have $p\tau > g$, but this is excluded when $g - p\tau>0$.

\subsection{Example}
\noindent
{In order to investigate the effectiveness of on-off traffic shaping on privacy, we conducted the following test: we collected packet timestamp traces from 10 different web sites.  The traces are of approximately same length.  We compared all of the samples with each other using Dynamic Time Warping and calculated their $F$-distance \cite{feghhi15}. The $F$-distance, which is a value in $[0,1]$, is a measure of the similarity between two timestamp traces and is used in \cite{feghhi15} as the basis for a successful timing-only traffic analysis attack based in $k$-NN clustering.  The $F$-distance is smaller when traces are similar and increases as they become more different.  It can be seen in Figure \ref{fig:rates} that the average distance between the traces for different web sites is relatively large, which would enable an attacker to distinguish between them.}

\begin{figure}
\centering
	\includegraphics[trim=0cm 1cm 0cm 1cm,clip,width=0.95\columnwidth]{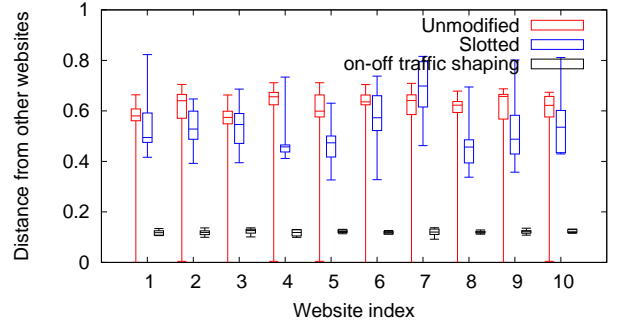}
	\caption{Comparison of $F$-distance between different web traces for unmodified, slotted and queued+slotted scenarios.	The slot size is $0.01$s, $g=5$, $\tau=10$ and the window size for evaluating $F$-distance is 0.2.}
	\label{fig:rates}
\end{figure}

{We then manually divided time into slots (of 10ms duration) and adjusted the packet timings so that they are only sent at the beginning of a slot\footnote{Note that using an actual time slotted tunnel adds an additional distortion due to network protocols. So by manually time slotting the time we are considering a slightly easier scenario for the attacker.}.  It can be seen in Figure \ref{fig:rates} that time slotting decreases the average  distance between packet traces, the decrease is relatively small and would still allow an attacker to distinguish between different the traces for different web sites.}

{Finally we applied on-off traffic shaping and calculated the $F$-distance.  It can be seen that there is now a considerable drop in the average distance, and also in the variance.  This indicates that the modified traces are now much more similar, making it difficult for the attacker to distinguish between web sites.}

\begin{figure}
\centering
\subfloat[Duration]{
  \includegraphics[width=0.48\columnwidth]{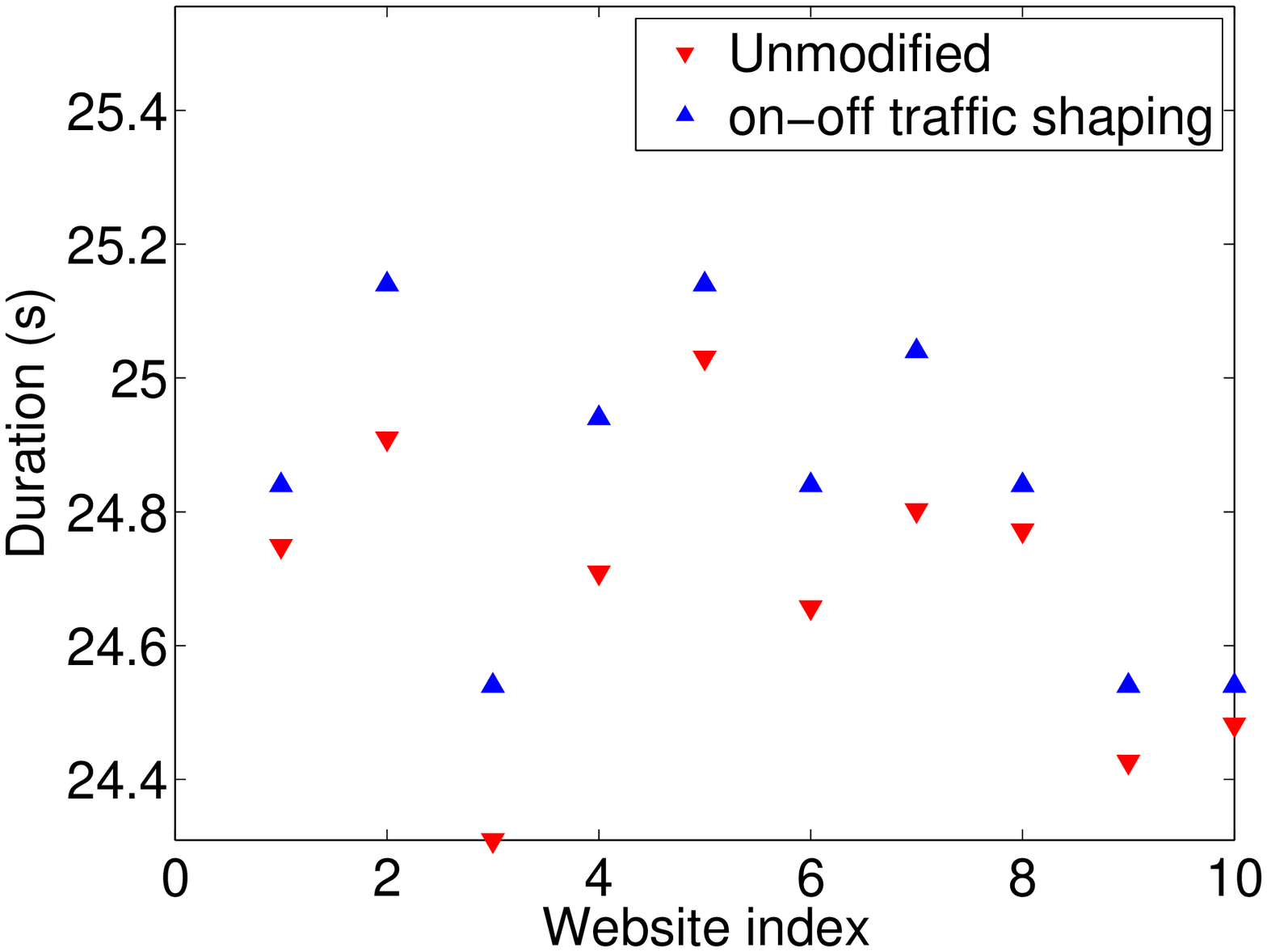}
	\label{fig:delay2}
}
\subfloat[Packet Count]{
  \includegraphics[width=0.48\columnwidth]{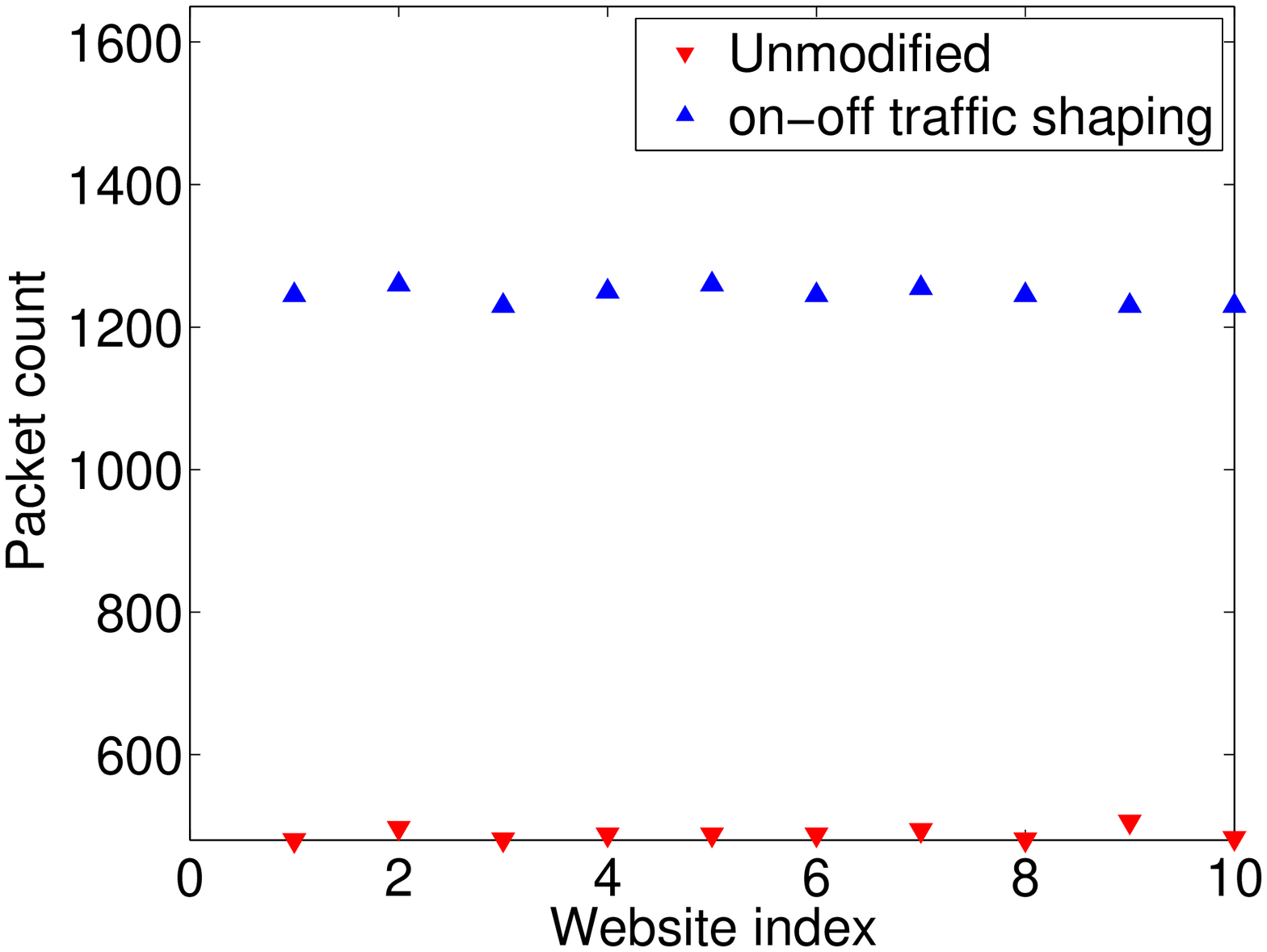}
	\label{fig:dummy2}
} \caption{{Comparison between transmission duration and packet count for unmodified and on-off shaped web traces. $g/\tau=0.5$.}}
\label{fig:penalty_2}
\end{figure}

\begin{figure}
\centering
\subfloat[Duration]{
  \includegraphics[width=0.48\columnwidth]{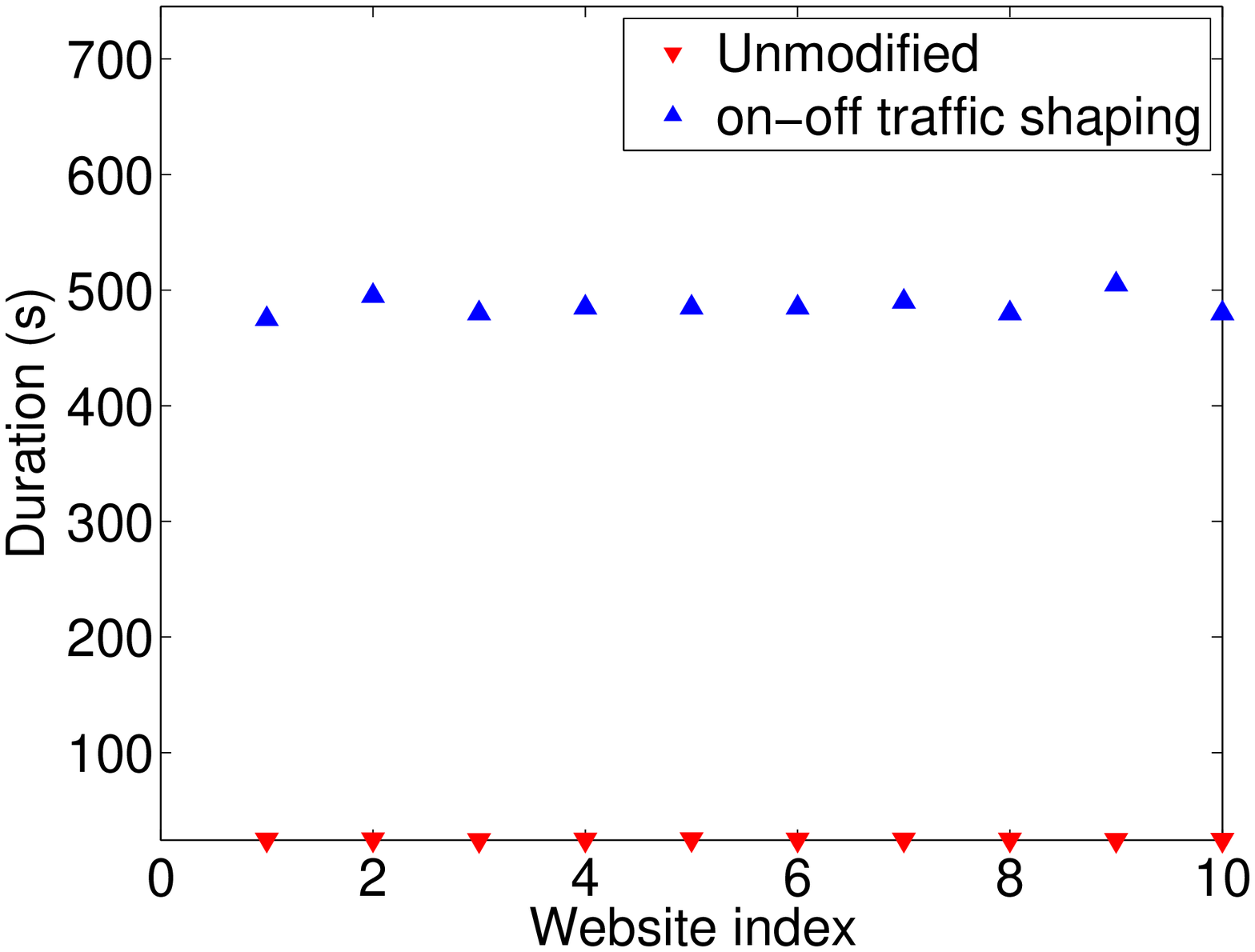}
	\label{fig:delay100}
}
\subfloat[Packet Count]{
  \includegraphics[width=0.48\columnwidth]{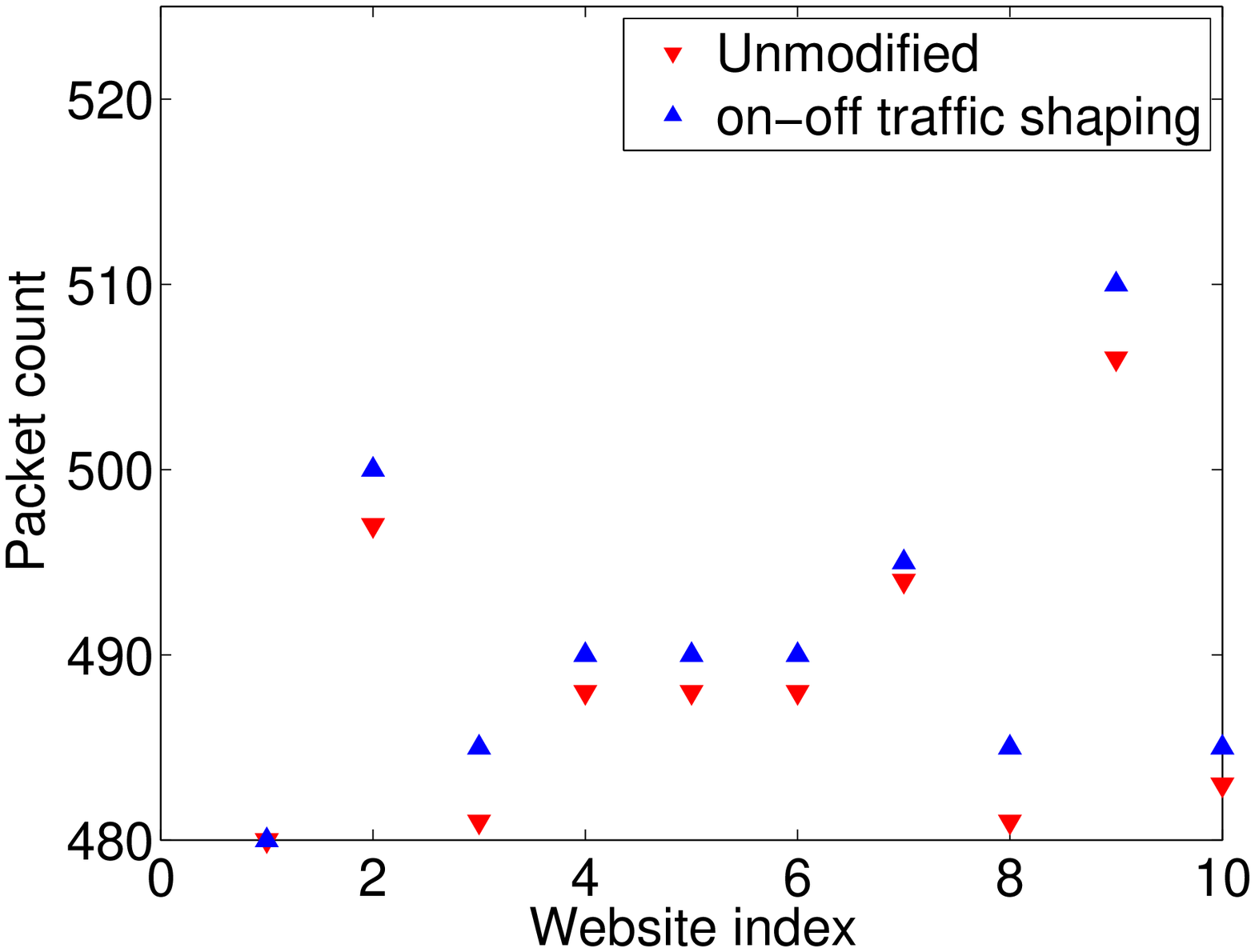}
	\label{fig:dummy100}
} \caption{{Comparison between transmission duration and packet count for unmodified and on-off shaped web traces. $g/\tau=0.01$.}}
\label{fig:penalty_100}
\end{figure}
{As already noted, traffic shaping introduces a queueing delay and the transmission of additional dummy packets. Figure \ref{fig:penalty_2} shows the measured delay and number of dummy packets when $g/\tau=0.5$.  It can be seen that the delay is negligible since $\tau$ is small but the number of dummy packets is almost 2 times that of the real traffic.  As $g/\tau$ is decreases, the number of dummy packets transmitted falls, but the delay increases.  This is illustrated in Figure \ref{fig:penalty_100} for $g/\tau = 0.01$. 
}


 


\section{Proportional Fair Privacy}
\label{sec:multiuser}
\noindent
We consider a shared network where each user applies traffic shaping before transmitting their traffic over the network, e.g. see Figure \ref{fig:intro}.   Note that traffic shaping is applied individually to each user's traffic, which is required to prevent active traffic injection attacks (where an adversary sharing a queue with a user injects traffic in a way that reveals the presence/absence of a user packet in the queue \cite{kadloor13}).   In practice this setup might correspond to a VPN where the client shapes and encrypts arriving traffic before transmitting it across the Internet to a VPN gateway.   As already noted, buffering and shaping may introduce delay (via buffering) and consume bandwidth (via dummy packet transmissions) and there is a joint trade-off between privacy, delay and throughput for the users sharing the network.

\subsection{Network Model}
\noindent
We consider a multi-user network serving a set of $\N$ users and a set $\F$ of flows.   Each flow $f\in\F$ has a source $s_f \in \N$, a mean offered load rate $\mu_f$ and a mean delay deadline $\sigma_f$.    The network uses a scheduler that can service the offered load provided the aggregate flow usage satisfies
\begin{align}
\sum_{f\in\F} \frac{\mu_f}{\psi_f} \le 1 
\end{align}
where $\psi_f$ is the physical transmit rate for flow $f$ and so $\mu_f/\psi_f$ is the mean fraction of airtime used by flow $f$ and for arrival rate $p_f$ and dummy rate $d_f$ we have $\mu_f = p_f + d_f$.

%
%

\subsection{Convexity of Traffic Shaping Delay}

As already noted, on-off traffic shaping introduces additional delay.   When on-time $g=1$ we have
\begin{align}
E(q_x) \approx \max\{ \frac{(2p - c)(1-p)}{2(c- p)}, 0\}
\end{align}
and mean waiting time
\begin{align}
w & = \frac{1-c}{1-p} \bigg[ \frac{E(q_x)}{p} + \frac{1}{2c} \bigg] \label{eq:delay2}
\end{align}
where $c=g/t > p$.

\begin{lemma}[Convexity of Delay]
\label{lm:convexity}
The mean waiting time $w$ in (\ref{eq:delay2}) is convex in the arrival rate $p$ and dummy rate $d= c-p$, although not jointly convex, for $p\in[0,1]$, $c\in(p,1]$.
\end{lemma}
\begin{proof}
When $c > 2p$ (so $E(q_x) = 0$), the second derivatives of $w$ with respect to $p$ and $c$ are
\begin{align*}
w_{pp} &= \frac{1-c}{c(1-p)^3},\ 
w_{cc} = \frac{1}{c^3(1-p)}\, 
w_{pc} = -\frac{1}{2c^2(1-p)2}
\end{align*}
It can be seen that $w_{pp}>0$ and $w_{cc}>0$ since $p,c\in[0,1]$.
%
%
When $c < 2p$ (so $E(q_x)>0$), we have
\begin{align}
w_{pp} & = (1-c) \bigg[ \frac{1}{(c-p)^3} - \frac{1}{p^3} + \frac{1}{c(1-p)^3} \bigg] \\
w_{cc} & = \frac{1-p}{(c-p)^3} + \frac{1}{c^3(1-p)} \\
w_{pc} & = -\frac{1}{2} \bigg[ \frac{2 - p - c}{(c-p)^3} + \frac{1}{p^2} + \frac{1}{c^2(1-p)^2} \bigg]
\end{align}
Since $c > p$ then $w_{cc}>0$.  Since $c < 2p$, $w_{pp}>0$.  However the Hessian need not be positive semidefinite (for example when $p=0.5$ and $c = 0.9$).
Since $w_{pp}>0$ and $w_{cc}>0$, the delay $w$ is convex in $p$ and $c$ individually, but since the Hessian is not positive semidefinite $w$ is not jointly convex in these variables.
%
Now $d= c-p$ is a linear function of $c$ and $p$.  Convexity is preserved under linear transformations and so the stated result follows.
\end{proof}



Note that since the on time $g$ and cycle time $\tau$ are expressed as numbers of slots, they are integer valued $g,\tau\in\mathbb{N}$ and so the domain of ratio $c$ is the rational numbers $\mathbb{Q}$ rather than the real-valued numbers.   We therefore consider the relaxed problem where $c$ takes values in $(p,1]\subset\mathbb{R}$ in order to ensure convexity.    This relaxation does not, however, entail any loss of generality.   Rather than using fixed on time $g\in\mathbb{N}$, define a sequence $g_k\in\mathcal{N}$, $k=1,2,\cdots$.   For a given value of $c\in\mathbb{R}$, by queue continuity an integer-valued sequence $g_k$ exists such $|\sum_k g_k-\hat{g}|$ is bounded and so the waiting time can be made arbitrarily close to that with specified $\hat{g}=c\tau\in\mathbb{R}$.

\subsection{Proportional Fair Allocation}
\label{sec:optimization}
\noindent
We consider the following optimisation $P$ which maximises log-rate (and so is proportional fair) subject to network and traffic shaping delay constraints:
\begin{align}
\min_{\mathbf{p},\mathbf{d}} & \quad U(\mathbf{p}):=\sum_{f \in \F} -\log(p_f) \label{eq:optimization}\\
\text{s.t.} & \quad \mathbf{w} \preceq \boldsymbol{\sigma} \label{eq:delaycons}\\
& \quad \sum_{f \in \F} \frac{p_f + d_f}{\psi_f} \leq 1 \label{eq:netcons} \\
& \quad \mathbf{p},\mathbf{d}  \in [0,1]^{|\F|} \label{eq:dummycons}
\end{align}
with
\small
\begin{align}
w_f = \frac{1-(p_f+d_f)}{2(1-p_f)}\bigg[\max\{\frac{(p_f-d_f)(1-p_f)}{p_fd_f},0\} + \frac{1}{(p_f + d_f)} \bigg]
\end{align}
\normalsize
for all flows $f \in \F $ and where $\mathbf{w}=[w_f]$, $\boldsymbol{\sigma}=[\sigma_f]$, $\mathbf{p}=[p_f]$, $\mathbf{d}=[d_f]$ for $f \in \F$ are vectors in $\mathbb{R}^{|\F|}$.  Constraints (\ref{eq:delaycons}) impose the requirement that flow delay deadlines are met, while (\ref{eq:netcons}) ensures that the flow rates (including both information and dummy transmissions) can be scheduled by the network.

\subsection{Solving Non-Convex Optimisation $P$}
\noindent
Optimisation $P$ is not convex because, by Lemma \ref{lm:convexity}, the delay constraints are not jointly convex in $\mathbf{p}$ and $\mathbf{d}$.  Nevertheless, these constraints are convex in $\mathbf{p}$ and $\mathbf{d}$ individually.   This suggests the use of an alternating approach to solve $P$.  Namely, solve for $\mathbf{p}$ holding $\mathbf{d}$ constant, then solve for $\mathbf{d}$ holding $\mathbf{p}$ constant.   Let $\mathbf{p}^*_k$, $\mathbf{d}^*_k$, $k=1,2,\cdots$ denote the sequence of alternating solutions found in this way.   Each solution is feasible for problem $P$.  Further, since each individual optimisation is convex, we can find a global minimum and so $U(\mathbf{p}^*_{k+1}) \le U(\mathbf{p}^*_k)$.   Hence, the sequence $\mathbf{p}^*_k$, $\mathbf{d}^*_k$ is guaranteed to converge to a feasible stationary point of problem $P$.   While this stationary point is, in general, sub-optimal, in practice we have found that it is usually close to a global optimum.

To carry out each optimisation we use a subgradient method for simplicity and because of its suitability for distributed implementation.   Of course other methods might also be used.  The resulting procedure is summarised as follows:
\begin{algorithm}
\caption{Alternating Solution Method}
\label{alg:subgradient}
\begin{algorithmic}[1]
\Statex {\textbf{iterate $s$:}}
\Statex \quad {\textbf{iterate $t$:}}
	\Statex \hspace{\algorithmicindent} $\mathbf{p}(t+1) = \mathbf{p}(t) - \alpha \partial_{\mathbf{p}} L_{\mathbf{p}}(\mathbf{p}(t),\mathbf{d}(s-1),\boldsymbol{\lambda}_p(t))$
	\Statex \hspace{\algorithmicindent} $\boldsymbol{\lambda}_p(t+1) = [ \boldsymbol{\lambda}_p(t) + \alpha \partial_{\boldsymbol{\lambda}}L_{\mathbf{p}}(\mathbf{p}(t),\mathbf{d}(s-1),\boldsymbol{\lambda}_p(t))]^+$
\Statex \quad {\textbf{loop}}	
\Statex \quad {\textbf{iterate $t$:}}
	\Statex \hspace{\algorithmicindent} $\mathbf{d}(t+1) = \mathbf{d}(t) - \alpha \partial_{\mathbf{d}}L_{\mathbf{d}}(\mathbf{p}(s),\mathbf{d}(t),\boldsymbol{\lambda}_d(t))$
	\Statex \hspace{\algorithmicindent} $\boldsymbol{\lambda}_d(t+1) = [\boldsymbol{\lambda}_d(t) + \alpha \partial_{\boldsymbol{\lambda}} L_{\mathbf{d}}(\mathbf{p}(s),\mathbf{d}(t),\boldsymbol{\lambda}_d(t))]^+$
\Statex \quad {\textbf{loop}}	
\Statex {\textbf{loop}}
\end{algorithmic}
\end{algorithm}

with Lagrangians,
\begin{align*}
L_{\mathbf{p}}(\mathbf{p},\mathbf{d}^*,\boldsymbol{\lambda}) = & \sum_{f\in\F} -\log(p_f) + \sum_{f \in \F} \lambda_{1,f}(w_f - \sigma_f) + \\
& \lambda_2((\sum_{f\in\F}p_f + d^*_f) -1) + \\
& \sum_{f\in\F} \lambda_{3,f} (p_f - 1) - \sum_{f\in\F} \lambda_{4,f} p_f
\end{align*}
\begin{align*}
L_{\mathbf{d}}(\mathbf{p}^*,\mathbf{d},\boldsymbol{\lambda}) = & \sum_{f \in \F} \lambda_{1,f}(w_f - \sigma_f) + \\
& \lambda_2((\sum_{f\in\F}p^*_f + d_f) -1) - \sum_{f\in\F} \lambda_{5,f} d_f
\end{align*}


%
\subsection{Examples}
\label{sec:results}
\noindent
{We present a number of examples to illustrate the proportional fair allocation in a fully private shared network.}

\begin{figure}[!t]
\centering
\subfloat[Optimal throughputs]{
  \includegraphics[width=0.48\columnwidth]{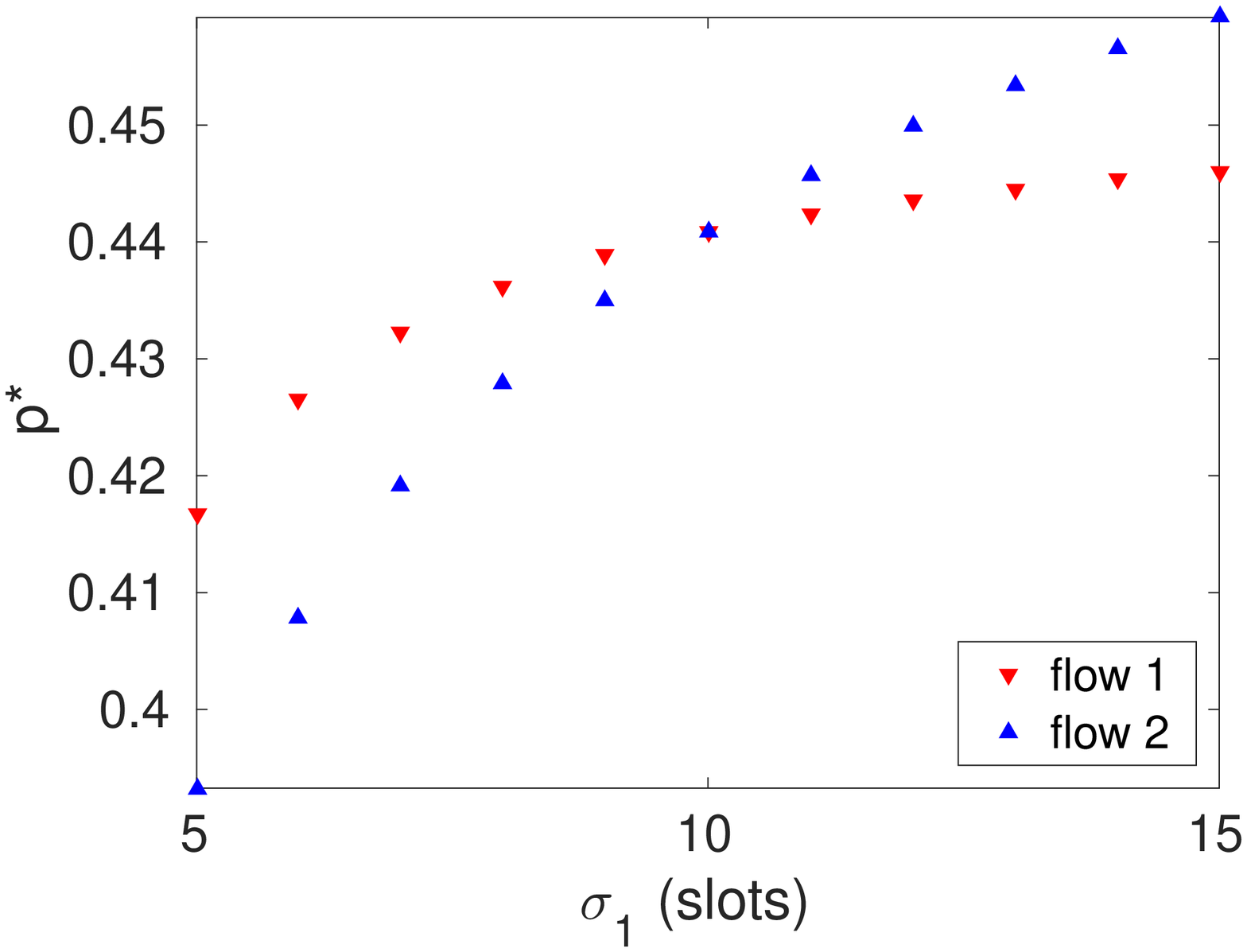}
	\label{fig:ex1optp}
}
\subfloat[Optimal Dummy rates]{
  \includegraphics[width=0.48\columnwidth]{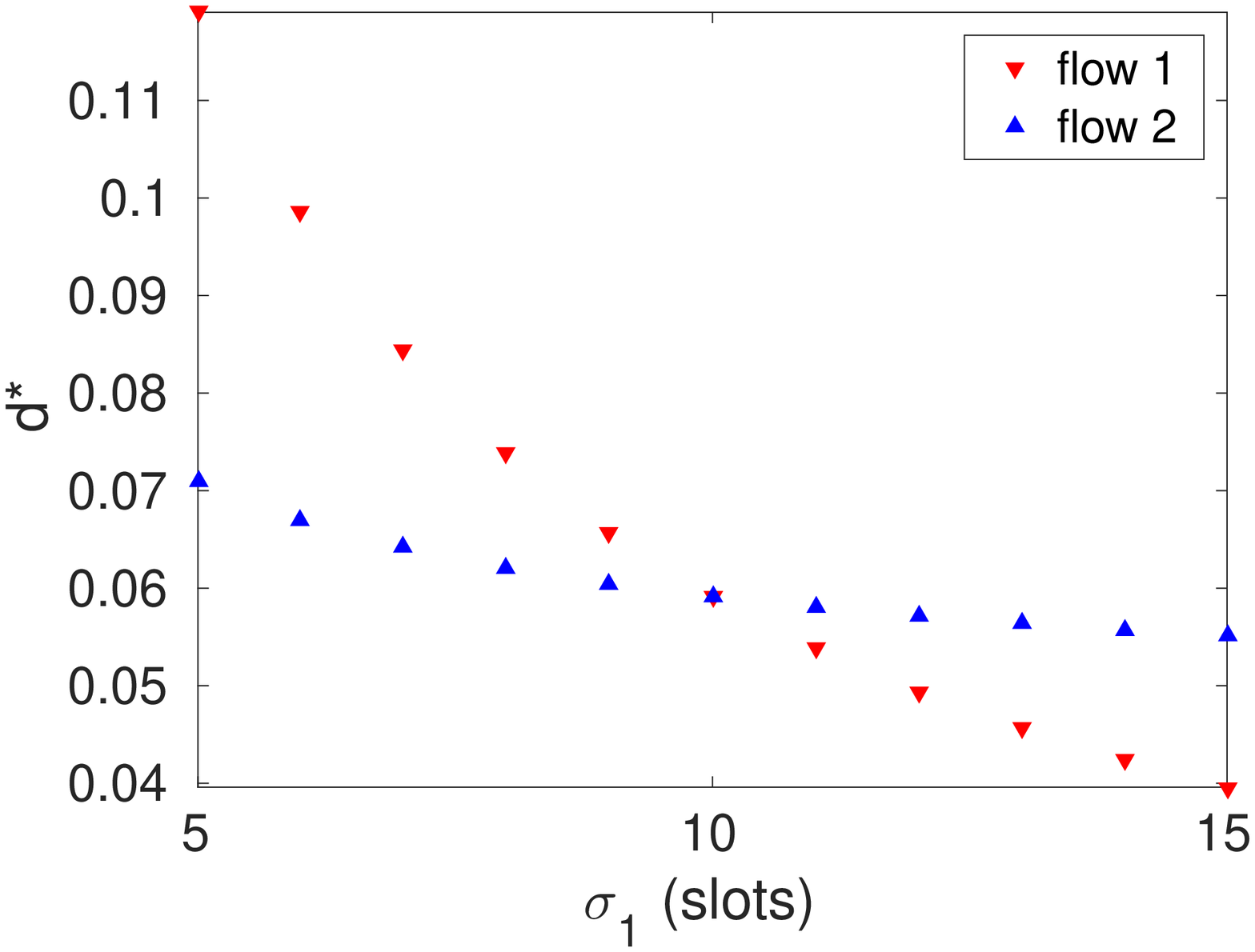}
	\label{fig:ex1optd}
} \caption{{Illustrating optimal throughput and dummy rate for delay deadlines $\sigma_2 = 10$ and $\sigma_1 = \{5, \dots, 15\}$.}}
\label{fig:example1}
\end{figure}

{\textit{Example 1: Fully private network.} Consider a network with two users having different delay deadlines.	We let delay deadline $\sigma_2 = 10$ and vary $\sigma_1$ between $\{\sigma_2 - 5, \sigma_2 + 5\}$ to observe the impact of the delay deadline on users' network share in a fully private network. Figure \ref{fig:example1} shows the proportional fair $p^*$ and $d^*$ vs $\sigma_1$. It can be seen that users' throughput and dummy rate are proportional to their delay deadlines \ie users with lower deadline are allowed to transmit more real and dummy traffic resulting a larger network share.}

\textit{Example 2: Mix of private and non-private flows.} As already noted, users in a shared network need to sacrifice throughput and/or delay to achieve full privacy. However, staying within our traffic shaping framework, a user can choose to ignore privacy by sending no dummy packets and using  the full cycle length for transmitting information packets.  Of~course this allows an adversary to see their packet arrivals. The optimization problem for this scenario is similar to \ref{eq:optimization} except that the delay and dummy rate constraints are now
\begin{align*}
w_f \leq \sigma_f, & \quad f \in \F_{private} \\
d_f > 0, & \quad f \in \F_{private} \\
d_f = 0, & \quad f \in \F - \F_{private}
\end{align*}
where $\F_{private} \subset \F$ is the set of private users.

{In this example we consider similar conditions to those  in Example 1, but now User 2 is non-private. Figure \ref{fig:example3} shows $p^*$ and $d^*$ vs the delay deadline of User 1.  It can be seen, that User 1 consistently gets higher throughput than in a fully private network (compare Figures \ref{fig:ex3optp} and \ref{fig:ex1optp}).}

\begin{figure}[!t]
\centering
\subfloat[Optimal throughputs]{
  \includegraphics[width=0.48\columnwidth]{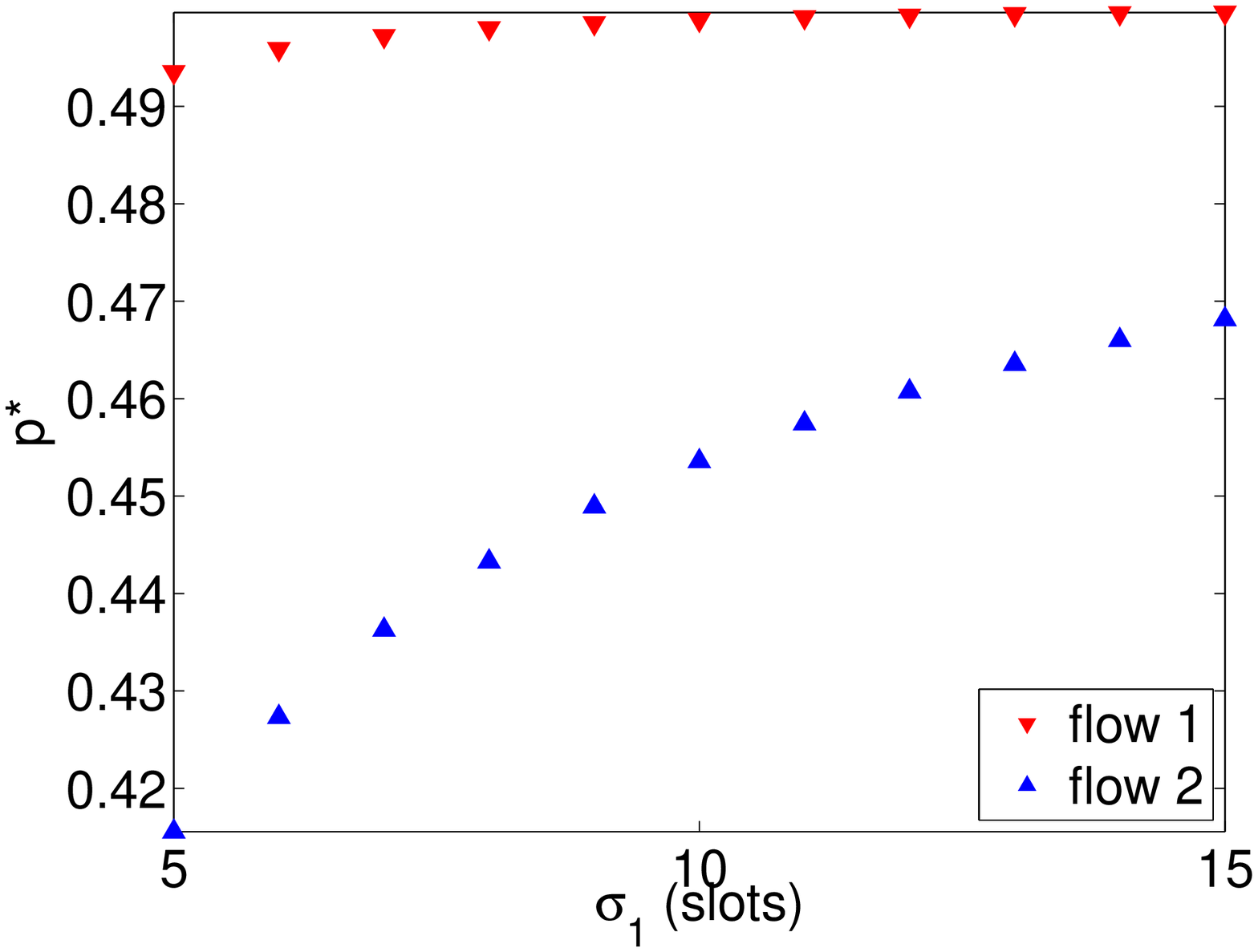}
	\label{fig:ex3optp}
}
\subfloat[Optimal dummy rates]{
  \includegraphics[width=0.48\columnwidth]{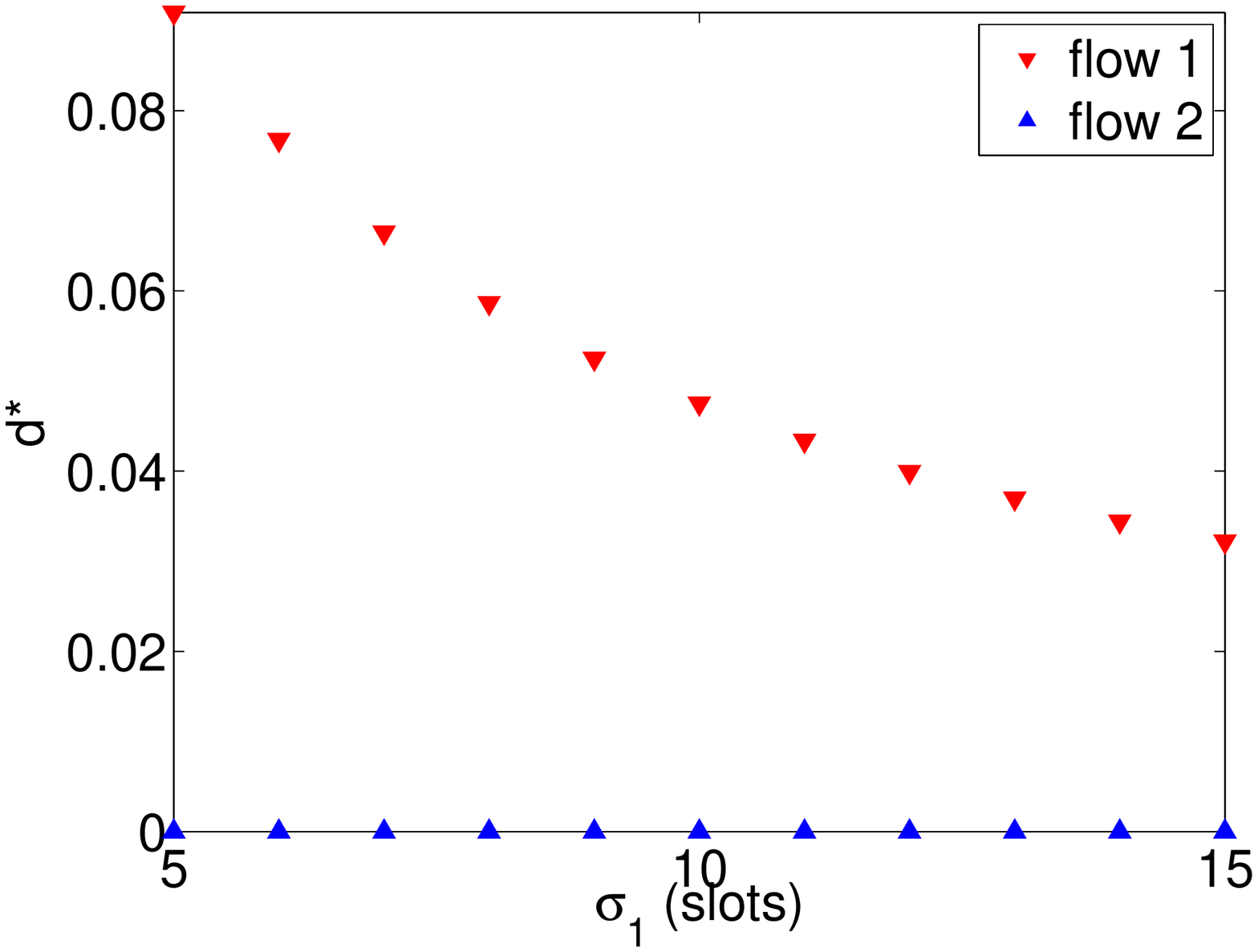}
	\label{fig:ex3optd}
} \caption{{Illustrating optimal throughput and dummy rate for a mix of private and non-private users. User 1 is private with $\sigma_1 = \{5, \dots, 15\}$ while User 2 ignores privacy by transmitting no dummy packets and using  the full cycle length for transmitting information packets.}}
\label{fig:example3}
\end{figure}

\section{Summary and Conclusions}
\noindent
{In this paper, we introduce a rate allocation scheme for private shared networks. First, a defence is proposed against timing only traffic analysis attacks which protects the user by transforming their packet arrival time sequence into one which contains no information about the packet arrival pattern of the original sequence. The transformation however imposes a delay on transmission and consumes bandwidth by transmitting dummy traffic. We address a shared network scenario where the performance of one user can affect the network experience of another. This leads to a further analysis of the resulting trade-off between user privacy and quality of experience and to the design of a proportional fair rate allocation algorithm.
}

\bibliography{references}{}
\bibliographystyle{plain}

\end{document}